\documentclass[runningheads]{llncs}
\usepackage{times}
\usepackage{cite}
\usepackage{listings}
\usepackage{colonequals}
\usepackage{xspace}
\usepackage{amssymb}
\usepackage{graphicx}
\usepackage{stmaryrd}
\usepackage{cite}
\usepackage{hyperref}
\usepackage{framed}
\usepackage{comment}
\usepackage{version}
\usepackage{url}

\usepackage{xcolor}

\usepackage{amsmath}    % need for subequations
\usepackage{booktabs}
\usepackage{url}
\usepackage{defs}
\usepackage{program}
\urldef{\mailsa}\path|firstname.lastname@epfl.ch|

\excludeversion{shortv}
\includeversion{longv}

\newcommand{\esolver}{\textsc{ESolver}\xspace}
\newcommand{\cvc}{\textsc{cvc}{\small 4}\xspace}

\newcommand{\teq}{\approx}

\newcommand{\M}{\mathcal{M}}

\newcommand{\rcon}{\ensuremath{\mathrm{rcon}}}

\newcommand{\define}[1]{\textsl{#1}}

\newcommand{\rem}[1]{\textcolor{magenta}{[#1]}}

\newcommand{\ct}[1]{\rem{#1 --ct}}

\newcommand{\sorts}{\mathbf{S}}

\newcommand{\vars}{\mathbf{X}}

\newcommand{\I}{\mathcal{I}}

\newcommand{\mods}{\mathbf{I}}
\newcommand{\lan}{\mathbf{L}}
\newcommand{\props}{\mathbf{P}}

\newcommand{\ent}[1][]{\models_{#1}}
\newcommand{\tent}{\ent[T]}

\newcommand{\ssorts}[1]{#1^\mathrm{s}}
\newcommand{\sfuns}[1]{#1^\mathrm{f}}

\newcommand{\con}[1]{\mathsf{#1}}
\newcommand{\Bool}{\con{Bool}}
\newcommand{\Int}{\con{Int}}
\newcommand{\ite}{\con{ite}}
\newcommand{\ev}{\con{ev}}
\newcommand{\size}{\con{size}}

\newcommand{\TD}{T_\mathrm{D}}

\newcommand{\ltrue}{\top}
\newcommand{\lfalse}{\bot}

\newcommand{\nmf}[1]{#1\!\downarrow}

\begin{document}
\sloppy

\mainmatter  % start of an individual contribution

\newcommand{\mytitle}{On Counterexample Guided Quantifier Instantiation \\
for Synthesis in CVC4\thanks{
This work is supported in part by the European
    Research Council (ERC) Project \emph{Implicit Programming} and Swiss National Science Foundation Grant \emph{Constraint Solving Infrastructure for Program Analysis}.}
\thanks{
This paper is dedicated to the memory of Morgan Deters
who died unexpectedly in Jan 2015.
}}

\title{\mytitle}

\titlerunning{\mytitle}

% If you are not Chinese, who cares anyways.
%
\author{Andrew Reynolds\inst{1} \and 
Morgan Deters\inst{2} \and \\
Viktor Kuncak\inst{1} \and
Cesare Tinelli\inst{3} \and 
Clark Barrett\inst{2} }

    \renewcommand{\topfraction}{0.95}    % max fraction of floats at top
    \renewcommand{\bottomfraction}{0.95} % max fraction of floats at bottom
    %   Parameters for TEXT pages (not float pages):
    \setcounter{topnumber}{2}
    \setcounter{bottomnumber}{2}
    \setcounter{totalnumber}{4}     % 2 may work better
    \setcounter{dbltopnumber}{2}    % for 2-column pages
    \renewcommand{\dbltopfraction}{0.95} % fit big float above 2-col. text
    \renewcommand{\textfraction}{0.07}  % allow minimal text w. figs
    %   Parameters for FLOAT pages (not text pages):
    \renewcommand{\floatpagefraction}{0.7}      % require fuller float pages
        % N.B.: floatpagefraction MUST be less than topfraction !!
    \renewcommand{\dblfloatpagefraction}{0.7}   % require fuller float pages    

\institute{
{\'E}cole Polytechnique F{\'e}d{\'e}rale de Lausanne (EPFL), Switzerland
\and 
Department of Computer Science, New York University 
\and 
Department of Computer Science, The University of Iowa 
}

%\toctitle{Lecture Notes in Computer Science}
%\tocauthor{Authors' Instructions}
\maketitle

%at least 70 and at most 150 words
\begin{abstract}
We introduce the first program
synthesis engine implemented inside an SMT solver. 
%We formulate our technique as support for \emph{synthesis conjectures}.  
We present an approach that extracts solution functions from unsatisfiability 
proofs of the negated form of synthesis conjectures. 
We also discuss novel counterexample-guided techniques for quantifier instantiation
that we use to make finding such proofs practically feasible.
A particularly important class of specifications are single-invocation properties,
for which we present a dedicated algorithm.
To support syntax restrictions on generated
solutions, our approach can transform a solution
found without restrictions into the desired syntactic form.
As an alternative, we show how to 
use evaluation function axioms to embed syntactic restrictions into 
constraints over algebraic datatypes, and then use
an algebraic datatype decision procedure to drive synthesis. 
Our experimental evaluation on syntax-guided synthesis benchmarks shows 
that our implementation in the CVC4 SMT solver is 
competitive with state-of-the-art tools for synthesis.
\end{abstract}

\section{Introduction}

The synthesis of functions that meet a given specification is a long-standing fundamental
goal that has received great attention recently. This functionality directly applies
to the synthesis of functional programs \cite{KuncakETAL13FunctionalSynthesisLinearArithmeticSets,KuncakETAL12SoftwareSynthesisProcedures} but also translates to imperative programs through techniques that include bounding input space, verification condition generation, and invariant discovery \cite{SolarLezama13ProgramSketching,SolarLezamaETAL06CombinatorialSketchingFinitePrograms,
SrivastavaGulwaniFoster13TemplatebasedProgramVerificationProgramSynthesis}. 
Function synthesis is also an important subtask in the synthesis of protocols and reactive systems, especially when these systems are infinite-state \cite{AlurETAL14SynthesizingFinitestateProtocolsFromScenariosRequirements, RyzhykETAL14UserguidedDeviceDriverSynthesis}.
The SyGuS format and competition \cite{AlurETAL13SyntaxguidedSynthesis,
  AlurETAL2014SyGuSMarktoberdorf,  DBLP:journals/corr/RaghothamanU14}
   inspired by the success of the SMT-LIB and
SMT-COMP efforts \cite{BarrettETAL136YearsSmtcomp},
has significantly improved and simplified the process of rigorously comparing different solvers on synthesis problems.

Connection between synthesis and theorem proving was established already in
early work on the subject~\cite{MannaWaldinger80DeductiveApproachToProgramSynthesis,
Green69ApplicationTheoremProvingToProblemSolving}.
It is notable that early research \cite{MannaWaldinger80DeductiveApproachToProgramSynthesis} found that the capabilities of theorem provers were the main bottleneck for synthesis.
Taking lessons from automated software verification, 
recent work on synthesis has
made use of advances in theorem proving, particularly in SAT and SMT solvers. However, that work
avoids formulating the overall synthesis task as a theorem proving problem directly. Instead, 
existing work typically builds custom loops outside of an SMT or SAT solver, often using numerous variants of counterexample-guided synthesis. A typical role of the SMT
solver has been to validate candidate solutions and provide counterexamples
that guide subsequent search, although approaches such as symbolic term
exploration \cite{KneussETAL13SynthesisModuloRecursiveFunctions} also use an
SMT solver to explore a representation of the space of solutions. In existing
approaches, SMT solvers thus receive a large number of separate queries, with
limited communication between these different steps.

\sparagraph{Contributions.}
In this paper, we revisit the formulation of the overall synthesis task as a theorem proving problem.
We observe that SMT solvers already have some of the key functionality for synthesis; we 
show how to improve existing algorithms and introduce new ones to make SMT-based synthesis competitive. 
%The resulting implementation outperforms the state of the art substantially on important classes of synthesis problems. 
%Our specific contributions are the following.
Specifically, we do the following.
\begin{itemize}
\item We show how to formulate an important class of synthesis problems as the problem of disproving universally quantified formulas, and how to synthesize functions automatically from selected instances of these formulas.
\item We present counterexample-guided techniques for quantifier instantiation, which are crucial to obtain competitive performance on synthesis tasks.
\item We discuss techniques to simplify the synthesized functions, to help ensure that they are small and adhere to specified syntactic requirements.
\item We show how to encode syntactic restrictions using theories of algebraic datatypes and axiomatizable evaluation functions. 

%This results in a flexible procedure that has desirable theoretical properties
%and also shows promise in practice.
\item We show that for an important class of single-invocation properties, the synthesis of functions from relations, the implementation of our approach in CVC4 significantly outperforms leading tools from the SyGuS competition.
\end{itemize}

% If this gets accepted, we should cite Andy's TR 

\sparagraph{Preliminaries.} %\label{sec:prelim}
Since synthesis involves finding (and so proving the existence) of functions,
we use notions from many-sorted \emph{second-order} logic to define the general problem.
We fix a set $\sorts$ of \define{sort symbols} and 
an (infix) equality predicate $\teq$ of type 
$\sigma \times \sigma$ for each $\sigma \in \sorts$.
%, which we always interpret as the identity relation over 
%(the set denoted by) $\sigma$.
For every non-empty sort sequence $\vec \sigma \in \sorts^+$
with $\vec \sigma = \sigma_1 \cdots \sigma_n\sigma$,
we fix an infinite set $\vars_{\vec \sigma}$
of \define{variables $x^{\sigma_1 \cdots \sigma_n \sigma}$ 
of type $\sigma_1 \times \cdots \times \sigma_n \to \sigma$}.
%\footnote{
%Note that we do not consider types or order greater than the second,
%i.e., of the form
%$\tau_1 \times \cdots \times \tau_n \to \tau_{n+1}$ where each $\tau_i$ 
%can be itself a type.
%}
For each sort $\sigma$ we identity the type $() \to \sigma$ with $\sigma$ and
call it a \define{first-order type}.
We assume the sets $\vars_{\vec \sigma}$ are pairwise disjoint and 
let $\vars$ be their union.
A \define{signature} $\Sigma$ consists of 
a set $\ssorts{\Sigma} \subseteq \sorts$ of sort symbols and
%a set $\preds{\Sigma}$ of \define{(sorted) predicate symbols} $p^{S_1 \cdots S_n}$,
a set $\sfuns{\Sigma}$ of 
\define{function symbols $f^{\sigma_1 \cdots \sigma_n \sigma}$ 
of type $\sigma_1 \times \cdots \times \sigma_n \to \sigma$},
where $n \geq 0$ and $\sigma_1, \ldots, \sigma_n, \sigma \in \ssorts{\Sigma}$.
%When $n$ above is 0, we call $f$ a \define{constant symbol}.
We drop the sort superscript from variables or function symbols
when it is clear from context or unimportant.
We assume that signatures always include a Boolean sort $\Bool$ and constants 
$\ltrue$ and $\lfalse$ of type $\Bool$ (respectively, for true and false).
Given a many-sorted signature $\Sigma$ 
together with quantifiers and lambda abstraction,
the notion of well-sorted ($\Sigma$-)term, atom, literal, clause, and formula
with variables in $\vars$ are defined as usual in second-order logic.
%they are referred to respectively as \define{$\Sigma$-terms},  
%\define{$\Sigma$-atoms} and so on.
All atoms have the form $s \teq t$.
%with $s$ and $t$ of the same sort.
Having $\teq$ as the only predicate symbol causes no loss of generality
since we can model other predicate symbols as function symbols 
with return sort $\Bool$.
We will, however, write just $t$ in place of the atom $t \teq \ltrue$, 
to simplify the notation. 
A $\Sigma$-term/formula is \define{ground} if it has no variables,
it is \define{first-order} if it has only \define{first-order variables},
that is, variables of first-order type.
When $\vec{x} = (x_1,\ldots,x_n)$ is a tuple of variables and 
$Q$ is either $\forall$ or $\exists$,
we write $Q \vec x\, \varphi$ as an abbreviation of
$Q x_1 \cdots Q x_n\, \varphi$.
%
%Free and bound occurrences of a variable in a formula are also defined as usual.
%A \define{($\Sigma$-)sentence} is a $\Sigma$-formula with no free variables.
%We denote by $\vars(\varphi)$ 
%the set of variables occurring free in the formula $\varphi$.
%, and extend the notation to sets of formulas in the obvious way.
If $e$ is a $\Sigma$-term or formula and 
$\vec{x} = (x_1,\ldots,x_n)$ has no repeated variables,
we write %$\varphi[x_1,\ldots,x_n]$, or just $\varphi[\vec x]$, 
$e[\vec x]$
to denote that all of $e$'s free variables are from $\vec x$;
%if $t_i$ is a term of the same sort as $x_i$
%for $i=1,\ldots, n$, we write $\varphi[t_1,\ldots,t_n]$ 
if $\vec{t} = (t_1,\ldots,t_n)$ is a term tuple, 
we write $e[\vec t]$ for the term or formula obtained from $e$ 
by simultaneously replacing, for all $i=1,\ldots,n$, every occurrence 
of $x_i$ in $e$ by $t_i$.
%we write $\vec s \teq \vec t$ for the set
%$\{s_1 \teq t_1, \ldots, s_n \teq\ t_n\}$. 
%We denote finite tuples of elements by letters in bold font,
%and use comma (,) for tuple concatenation.
%
A \define{$\Sigma$-interpretation $\I$} %is a mathematical structure that 
maps:
each $\sigma \in \ssorts{\Sigma}$ to a non-empty set $\sigma^\I$,
the \define{domain} of $\sigma$ in $\I$, with $\Bool^\I = \{\ltrue, \lfalse\}$;
%%\footnote{
%%We will use the calligraphic letters $\I$, $\mathcal{B}$, $\ldots$ 
%%to denote interpretations, 
%%and the corresponding subscripted Roman letters 
%%$\I_S$, $B_S$, $\ldots$ to denote sort domains.
%%}
each
$u^{\sigma_1 \cdots \sigma_n\sigma} \in \vars \cup \sfuns{\Sigma}$ 
to a total function 
$u^\I : \sigma_1^\I \times \cdots \times \sigma_n^\I \rightarrow \sigma^\I$
when $n > 0$ and
to an element of $\sigma^\I$ when $n = 0$.
The interpretation $\I$ induces as usual a mapping from terms $t$ 
of sort $\sigma$ to elements $t^\I$ of $\sigma^\I$.
If $x_1, \ldots,x_n$ are variables and $v_1,\ldots,v_n$ are well-typed values 
for them, we denote by $\I[x_1 \mapsto v_1, \ldots, x_n \mapsto v_n]$
the $\Sigma$-interpretation that maps each $x_i$ to $v_i$ and is otherwise 
identical to $\I$.
A satisfiability relation between $\Sigma$-interpretations and 
$\Sigma$-formulas is defined inductively as usual.

A \define{theory} is a pair $T = (\Sigma, \mods)$ where 
$\Sigma$ is a signature and  $\mods$ is a non-empty class of $\Sigma$-interpretations,
the \define{models} of $T$,
that is closed under variable reassignment
(i.e., every $\Sigma$-interpretation that differs from one in $\mods$
only in how it interprets the variables is also in $\mods$) and isomorphism.
%$\I[x \mapsto a] \in \mods$ 
%for all $\I \in \mods$,
%all variables $x$ of sort $S$ and all $a \in \I_S$,
%where $\I[x \mapsto a]$ is the $\Sigma$-interpretation 
%that maps $x$ to $a$ and is otherwise identical to $\I$.
%An \define{$\lo$-formula} is an element of $\Fo$ and
%an \define{$\lo$-interpretation} is an element of $\mods$.
A $\Sigma$-formula $\varphi[\vec x]$ is 
\define{$T$-satisfiable} (resp., \define{$T$-unsatisfiable}) 
if it is satisfied by some (resp., no) interpretation in $\mods$.
A satisfying interpretation for $\varphi$ \define{models (or is a model of)} $\varphi$.
%
%A set $\Gamma$ of formulas \define{$T$-entails} a $\Sigma$-formula $\varphi$,
%written $\Gamma \tent \varphi$,
%if every interpretation in $\mods$ that satisfies all formulas in $\Gamma$
%satisfies $\varphi$ as well.
%The formula $\varphi$ is \define{valid} in $T$, written $\ent[T] \varphi$
%if $\emptyset \ent[T] \varphi$.
%The set $\Gamma$ is \define{$T$-satisfiable} if $\Gamma \not\tent \lfalse$.
%
%\ct{add definition of uninterpreted symbol}
A formula $\varphi$ is \define{$T$-valid}, written $\tent \varphi$,
if every model of $T$ is a model of $\varphi$.
Given a fragment $\lan$ of the language of $\Sigma$-formulas,
a $\Sigma$-theory $T$ is \define{satisfaction complete with respect to $\lan$}
if every $T$-satisfiable formula of $\lan$ is $T$-valid.
In this paper we will consider only theories that are satisfaction complete
wrt the formulas we are interested in.
Most theories used in SMT (in particular, all theories of a specific structure
such various theories of the integers, reals, strings, algebraic datatypes, 
bit vectors, and so on) are satisfaction complete with respect 
to the class of closed first-order $\Sigma$-formulas.
Other theories, such as the theory of arrays, are satisfaction complete only
with respect to considerably more restricted classes of formulas.
\section{Synthesis inside an SMT Solver}
%===============================================================================

We are interested in synthesizing computable functions automatically 
from formal logical specifications stating properties of these functions.
As we show later, under the right conditions, we can formulate a version 
of the synthesis problem in \emph{first-order logic} alone, 
%ct which is the starting point for allowing 
which allows us to tackle the problem using SMT solvers.

%\ct{Some of the content in the following, up to the first subsection, may 
%end up in the introduction.}

%\subsection{Synthesis Conjectures}

We consider the synthesis problem in the context of some theory $T$ of signature $\Sigma$
that allows us to provide the function's specification as a $\Sigma$-formula.
Specifically, we consider \define{synthesis conjectures} expressed as 
(well-sorted) formulas of the form
\begin{eqnarray} \label{eqn:syn_conj}
\exists f^{\sigma_1 \cdots\sigma_n\sigma} \:
\forall x_1^{\sigma_1} \: \cdots \: \forall x_n^{\sigma_n} \:
  P[f, x_1, \ldots, x_n]
\end{eqnarray}
or $\exists f\, \forall \vec x\, P[f, \vec x]$, for short,
where the second-order variable $f$ represents 
the function to be synthesized and $P$ is a $\Sigma$-formula encoding
properties that $f$ must satisfy for all possible values of the input tuple
$\vec x = (x_1,\ldots,x_n)$.
%We will say that the synthesis conjecture is \define{realizable} 
%if $(1)$ is $T$-satisfiable.
%
In this setting, finding a witness for this satisfiability problem amounts 
to finding a function of type
$\sigma_1 \times \cdots \times \sigma_n \to \sigma$ in some model of $T$
that satisfies $\forall \vec x\, P[f, \vec x]$.
%In general, the set of solutions to such problems is not even recursively enumerable.
Since we are interested in automatic synthesis, we the restrict ourselves here 
to methods that search over a subspace $S$ of solutions representable 
syntactically as $\Sigma$-terms. 
We will say then that a synthesis conjecture is \define{solvable} 
if it has a syntactic solution in $S$.

In this paper we present two approaches that work 
with classes $\lan$ of synthesis conjectures and $\Sigma$-theories $T$ 
that are satisfaction complete wrt $\lan$.
In both approaches, we solve a synthesis conjecture
$\exists f\, \forall \vec x\, P[f, \vec x]$
by relying on quantifier-instantiation techniques % for SMT 
to produce 
a first-order $\Sigma$-term $t[\vec x]$ of sort $\sigma$ such that 
$\forall \vec x\, P[t, \vec x]$ is $T$-satisfiable.
When this $t$ is found, the synthesized function is denoted
by $\lambda \vec x.\, t$\,.

In principle, to determine the satisfiability 
of $\exists f\, \forall \vec x\, P[f, \vec x]$
an SMT solver supporting the theory $T$ can consider the satisfiability 
of the (open) formula $\forall \vec x\, P[f, \vec x]$ by treating $f$ 
as an uninterpreted function symbol.
This sort of Skolemization is not usually a problem for SMT solvers
as many of them can process formulas with uninterpreted symbols.
The real challenge is the universal quantification over $\vec x$
because it requires the solver to construct internally (a finite representation of)
an interpretation of $f$ that is guaranteed to satisfy $P[f, \vec x]$
for every possible value of $\vec x$~\cite{GeDeM-CAV-09,ReyEtAl-CADE-13}.

%First, the solver must construct a stream of candidate interpretations for $f$ based on its partial model, 
%which by default gives no guarentee that an interpretation will eventually be discovered that satisifies this (or other) quantified formulas.
%Moreover, the solver must be extended with methods for determining when universally quantified formulas are satisfied.
%In fact, showing that a universally quantified formula is satisfied for all $\vec i$ often is accomplished by showing that 
%its negation under the candidate interpretation of $f$ is unsatisfiable~\cite{GeDeM-CAV-09}, which itself reduces to a ground satisfiability query.
%An alternative line of research in the domain of software verification
%has explored specialized techniques for establishing the satisfiability of quantified horn clauses~\cite{beyene2013solving, beyene2014constraint, DBLP:conf/sas/BjornerMR13} %~\cite{DBLP:conf/sas/BjornerMR13}, 
%which has had success for handling clauses in several theories.

%More traditionally, SMT solvers have focused on determining the unsatisfiability of inputs containing quantified formulas
%by adding instances of quantified formulas to their set of clauses until a refutation is found at the ground level.
%While these techniques are generally incomplete, 
%SMT solvers are often quite effective tools for doing so~\cite{DBLP:conf/cade/MouraB07, reynolds14quant_fmcad}~\cite{DBLP:conf/cade/MouraB07, reynolds14quant_fmcad}.

More traditional SMT solver designs to handle universally quantified formulas 
have focused on instantiation-based methods to show \emph{un}satisfiability.
They generate ground instances of those formulas until a refutation is found 
at the ground level~\cite{Detlefs03simplify:a}.
While these techniques are incomplete in general,
they have been shown to be quite effective in practice~\cite{MouraBjoerner07EfficientEmatchingSmtSolvers, reynolds14quant_fmcad}.
%Arguably, doing so is more natural for an SMT solver and poses fewer complications than establishing the satisfiability of quantified formulas.
For this reason, we advocate approaches to synthesis geared toward establishing 
the \emph{unsatisfiability of the negation} of the synthesis conjecture:
\begin{eqnarray} \label{eqn:neg_syn_conj}
  \forall f\,\exists \vec x\, \lnot P[f, \vec x]
\end{eqnarray}
Thanks to our restriction to satisfaction complete theories, 
(\ref{eqn:neg_syn_conj}) is $T$-unsatisfiable exactly when 
the original synthesis conjecture (\ref{eqn:syn_conj}) is $T$-satisfiable.\footnote{
Other approaches in the verification and synthesis literature 
also rely implicitly, and in some cases unwittingly, on this restriction
or stronger ones.
We make satisfaction completeness explicit here 
as a sufficient condition 
for reducing satisfiability problems to unsatisfiability ones.
}
Moreover, as we explain in this paper, a syntactic solution $\lambda x.\,t$ 
for (\ref{eqn:syn_conj}) can be constructed from a refutation 
of (\ref{eqn:neg_syn_conj}),
as opposed to being extracted from the valuation of $f$ in a model of $\forall \vec x\, P[f, \vec x]$.
%(\ref{eqn:syn_conj}).

\sparagraph{Two synthesis methods.}
Proving (\ref{eqn:neg_syn_conj}) unsatisfiable poses its own challenge 
to current SMT solvers, namely, dealing with the second-order universal 
quantification of $f$.
To our knowledge, no SMT solvers so far had direct support for higher-order quantification.
In the following, however, we describe two specialized methods to refute 
negated synthesis conjectures like (\ref{eqn:neg_syn_conj}) 
that build on existing capabilities of these solvers.

The first method applies to a restricted, but fairly common, case of synthesis 
problems $\exists f\, \forall\vec x\, P[f, \vec x]$
where every occurrence of $f$ in $P$ is in terms of the form $f(\vec x)$. 
In this case, we can express the problem in the first-order form 
$\forall \vec x. \exists y. Q[\vec x,y]$ and then tackle its negation using 
appropriate quantifier instantiation techniques.

The second method follows the \emph{syntax-guided synthesis} paradigm~
\cite{AlurETAL13SyntaxguidedSynthesis, AlurETAL2014SyGuSMarktoberdorf}
where the synthesis conjecture is accompanied by an explicit syntactic restriction
on the space of possible solutions.
%A recent line of research has targeted this class of problems~\cite{6679385}
%since it is of practical interest in various applications.
Our syntax-guided synthesis method is based on encoding the syntax of terms
as first-order values. We use
a deep embedding 
into an extension of the background theory $T$ 
with a theory of algebraic data types, encoding the restrictions of a syntax-guided
synthesis problem.
\medskip

{\em For the rest of the paper, we fix a $\Sigma$-theory $T$ and a class $\props$ 
of quantifier-free $\Sigma$-formulas $P[f,\vec x]$ such that
$T$ is satisfaction complete with respect to the class of synthesis conjectures
$\lan := \{\exists f\, \forall\vec x\, P[f, \vec x] \mid P \in \props \}$.
}

%================================================================================
\section{Refutation-Based Synthesis}
\label{sec:refutation-based}
%================================================================================

When axiomatizing properties of a desired function $f$ 
of type $\sigma_1 \times \cdots \times \sigma_n \to \sigma$, 
a particularly well-behaved class are \emph{single-invocation properties}
(see, e.g.,\ \cite{jacobs2011towards}). 
These properties include, in particular, standard function contracts, 
so they can be used to synthesize a function implementation given 
its postcondition as a relation between the arguments and the result 
of the function.
This is also the form of the specification for synthesis problems considered 
in complete functional synthesis \cite{KuncakETAL10CompleteFunctionalSynthesis,KuncakETAL12SoftwareSynthesisProcedures,
KuncakETAL13FunctionalSynthesisLinearArithmeticSets}.
Note that, in our case, we aim to prove that the output exists for all inputs,
as opposed to, more generally, computing the set of inputs for which
the output exists.

A \define{single-invocation property} is any formula of the form
$Q[\vec x, f(\vec x)]$ obtained as an instance of a quantifier-free formula
$Q[\vec x, y]$ not containing $f$.
Note that the only occurrences of $f$ in $Q[\vec x, f(\vec x)]$ are 
in subterms of the form $f(\vec x)$ 
with the \emph{same} tuple $\vec x$ of \emph{pairwise distinct} variables.\footnote{
An example of a property that is \emph{not} single-invocation is
$\forall x_1\,x_2\, f( x_1, x_2 ) \teq f( x_2, x_1 )$.
%, stating that $f$ is a commutative function.
}
The conjecture 
%\vk{should we say \textbf{realizability} problem? Synthesis problem
%should return the representation of the function.}
$\exists f\, \forall \vec x\, Q[\vec x, f(\vec x)]$ is %equivalent in $T$ 
logically equivalent to the \emph{first-order} formula
\begin{equation} \label{eqn:syn_conj_no_syntax}
  \forall \vec x\, \exists y\, Q[\vec x, y]
\end{equation}
%
%In contrast to~\eqref{eqn:syn_conj}, formula~(\ref{eqn:syn_conj_no_syntax}) 
%is first-order.
By the semantics of $\forall$ and $\exists$,
finding a model $\I$ for it amounts (under the axioms of choice) to finding a function
$h:\sigma_1^\I \times \cdots \times \sigma_n^\I \rightarrow \sigma^\I$ 
such that for all $\vec s \in \sigma_1^\I \times \cdots \times \sigma_n^\I$,
the interpretation $\I[\vec x \mapsto \vec s, y \mapsto h(\vec s)]$ satisfies
$Q[\vec x, y]$.
This section considers the case when $\props$ consists of single-invocation 
properties and describes a general approach for determining the satisfiability 
of formulas like~(\ref{eqn:syn_conj_no_syntax}) while computing 
a syntactic representation of a function like $h$ in the process.
For the latter, it will be convenient to assume that the language 
of functions contains an if-then-else operator $\ite$ of type 
$\Bool \times \sigma \times \sigma \to \sigma$ for each sort $\sigma$,
with the usual semantics.

If~(\ref{eqn:syn_conj_no_syntax}) belongs to a fragment that admits 
quantifier elimination in $T$, such as the linear fragment of integer arithmetic,
determining its satisfiability can be achieved using an efficient method 
for quantifier elimination~\cite{Monniaux10QuantifierEliminationLazyModelEnumeration, Bjoerner10LinearQuantifierEliminationAsAbstractDecision}.
Such cases have been examined in the context of software synthesis~\cite{KuncakETAL12SoftwareSynthesisProcedures}.
Here we propose instead an alternative instantiation-based approach aimed 
at establishing the unsatisfiability of the \emph{negated} form 
of~(\ref{eqn:syn_conj_no_syntax}):
\begin{equation} \label{eqn:neg_syn_conj_no_syntax}
  \exists \vec x\, \forall y\, \lnot Q[\vec x, y]
\end{equation}
or, equivalently, of a Skolemized version $\forall y\, \lnot Q[\vec{\con k}, y]$ 
of ~(\ref{eqn:neg_syn_conj_no_syntax}) for some tuple $\vec{\con k}$ of fresh 
uninterpreted constants of the right sort.
Finding a $T$-unsatisfiable finite set $\Gamma$ of ground instances of 
$\lnot Q[\vec k, y]$, 
which is what an SMT solver would do to prove the unsatisfiability of~\eqref{eqn:neg_syn_conj_no_syntax}, 
suffices to solve the original synthesis problem.
The reason is that, then, a solution for $f$ can be constructed 
directly from $\Gamma$, as indicated by the following result.

\begin{proposition}\label{prop:ite-form}
\em
Suppose some set 
$\Gamma = \{\lnot Q[\vec{\con k}, t_1[\vec{\con k}]], \ldots, \lnot Q[\vec{\con k}, t_p[\vec{\con k}]]\}$ 
where $t_1[\vec x]$, $\ldots$, $t_p[\vec x]$ are $\Sigma$-terms of sort $\sigma$
is $T$-un\-satisfiable.
One solution for $\exists f\, \forall \vec x\, Q[\vec x, f(\vec x)]$ is
$\lambda \vec x.\, \ite( Q[\vec x, t_p], t_p, (\,\cdots\, \ite( Q[\vec x, t_2], t_2, t_1 ) \,\cdots\, ))$.
%
%Say that $\neg Q( t_1, \vec{\con k} ), \ldots, \neg  Q( t_n, \vec{\con k} ) \models_T \bot$ for fresh constants $\vec{\con k}$.
%Then, $\ell := \lambda \vec{\con k}. \ \ite( Q( t_1, \vec{\con k} ), t_1, \ldots \ite( Q( t_{n-1}, \vec{\con k} ), t_{n-1}, t_n ) \ldots )$ is a solution for $g$ in $\forall \vec i. \exists g. Q( g, \vec i )$.
\end{proposition}

\begin{longv}
\begin{proof}
Let $\ell$ be the solution specified above, and let $\vec u$ 
be an arbitrary set of ground terms of the same sort as $\vec x$.
Given a model $\I$, we show that $\I \models Q[ \vec u, \ell( \vec u ) ]$.
Consider the case that $\I \models Q[ \vec u, t_i [\vec{u}] ]$ for some $i \in \{ 2, \ldots, p \}$;
pick the greatest such $i$.
Then, $\ell( \vec u )^\I = ( t_i [\vec{u}] )^\I$, and thus $\I \models Q[ \vec u, \ell( \vec u ) ]$. 
If no such $i$ exists, then $\I \models \neg Q[ \vec u, t_i [\vec{u}] ]$ for all $i = 2, \ldots, p$, and $\ell( \vec u )^\I = ( t_1 [\vec{u}] )^\I$.
Since $\Gamma$ is $T$-unsatisfiable 
and $\vec{\con k}$ are fresh,
we have $\neg Q[ \vec u, t_2 [\vec{u}] ], \ldots, \neg  Q[ \vec u, t_p [\vec{u}] ] \models_T Q[ \vec u, t_1 [\vec{u}] ]$, which is $Q[ \vec u, \ell( \vec u ) ]$.
\eop
\end{proof}

\ 
\end{longv}

%In the following, we write $\lambda \vec{\con k}. \ite( Q( t_1, \vec{\con k} ), t_1, \ldots \ite( Q( t_{n-1}, \vec{\con k} ), t_{n-1}, t_n ) \ldots )$ shorthand as $L( t_1, \ldots, t_n )$.

\begin{example} \label{ex:max}
Let $T$ be the theory of linear integer arithmetic with the usual signature
and integer sort $\Int$.
%ct removed as requeste by viktor
%One can show that $T$ is satisfaction complete with respect to the whole class
%of closed $\Sigma$-formulas.
Let $\vec x = (x_1, x_2)$. 
Now consider the property
\begin{eqnarray} \label{eq:max-orig}
% P[f, x_1, x_2] :=
% f( x_1, x_2 ) \geq x_1 \land f( x_1, x_2 ) \geq x_2 \land 
% ( f( x_1, x_2 ) \teq x_1 \lor f( x_1, x_2 ) \teq x_2 )
 P[f, \vec x] :=
 f( \vec x ) \geq x_1 \land f( \vec x ) \geq x_2 \land 
 ( f( \vec x ) \teq x_1 \lor f( \vec x ) \teq x_2 )
\end{eqnarray}
with $f$ of type $\Int \times \Int \rightarrow \Int$ and
$x_1, x_2$ of type $\Int$.
The synthesis problem $\exists f\, \forall \vec x\, P[f, \vec x]$ is solved 
exactly by the function that returns the maximum of its two inputs.
Since $P$ is a single-invocation property, we can solve that problem by proving 
the $T$-unsatisfiability of the conjecture 
$\exists \vec x\, \forall y\, \lnot Q[\vec x, y]$ where 
\begin{eqnarray} \label{eq:max}
 Q[\vec x, y] & := & 
 y \geq x_1 \land y \geq x_2 \land ( y \teq x_1 \lor y \teq x_2 )
\end{eqnarray}
After Skolemization the conjecture becomes 
$\forall y\, \lnot Q[\vec{\con a}, y]$ 
for fresh constants $\vec{\con a} = (\con a_1, \con a_2)$.
When asked to determine the satisfiability of that conjecture an SMT solver
may, for instance, instantiate it with $\con a_1$ and then $\con a_2$ for $y$, 
producing the $T$-unsatisfiable set 
$\{\lnot Q[\vec{\con a}, \con a_1], \lnot Q[\vec{\con a}, \con a_2]\}$.
By Proposition~\ref{prop:ite-form}, one solution 
for $\forall \vec x\, P[f, \vec x]$ is
$f = \lambda \vec x.\, \ite( Q[\vec x, x_2], x_2, x_1 )$,
%which is $\lambda k_1 k_2. \ \ite( k_2 \geq k_1 \wedge k_2 \geq k_2 \wedge ( k_2 \teq k_1 \vee k_2 \teq k_2 ), k_2, k_1 )$ and simplifies 
which simplifies to $\lambda \vec x.\, \ite( x_2 \geq x_1, x_2, x_1 )$, representing the desired maximum function.
\qed
\end{example}

%-----------------------------------------------------------------------------
\begin{figure}[t]
\begin{enumerate}
%\begin{framed}
\item $\Gamma := \{\con G \Rightarrow Q[\vec{\con k}, \con e]\}$
where $\vec{\con k}$ consists of distinct fresh constants
\item Repeat
 \begin{itemize}
  \item[\ ]
  If there is a model $\I$ of $T$ satisfying $\Gamma$ and $\con G$ \\
  then let $\Gamma := \Gamma \cup \{ \lnot Q[\vec{\con k},t[\vec{\con k}]] \}$ 
  for some $\Sigma$-term $t[\vec x]$ such that $t[\vec{\con k}]^\I = {\con e}^\I$;  \\
  otherwise, return ``no solution found''
 \end{itemize}
until $\Gamma$ contains a $T$-unsatisfiable set 
$\{\lnot Q[\vec{\con k}, t_1[\vec{\con k}]], \ldots, \lnot Q[\vec{\con k}, t_p[\vec{\con k}]] \}$
\item 
Return $\lambda \vec x.\, \ite( Q[\vec x, t_p[\vec x]], t_p[\vec x],\ (\,\cdots\, \ite( Q[\vec x, t_2[\vec x]], t_2[\vec x], t_1[\vec x] ) \,\cdots\, ))$
for $f$
%If $p = 0$ then
%return $\lambda \vec x.\, t_1[\vec x]$ as a solution;\\
%otherwise, return $\lambda \vec x.\, \ite( Q[\vec x, t_1[\vec x]], t_1[\vec x],\ (\,\cdots\, \ite( Q[\vec x, t_p[\vec x]], t_p[\vec x], t_{p+1}[\vec x] ) \,\cdots\, ))$
%\end{framed}
\end{enumerate}
\vspace{-2ex}
\caption{A refutation-based synthesis procedure for single-invocation property 
$\exists f\, \forall \vec x\, Q[\vec x, f(\vec x)]$.
\label{fig:proc1}}
\end{figure}

\sparagraph{Synthesis by Counterexample-Guided Quantifier Instantiation.}
Given Proposition~\ref{prop:ite-form}, the main question is how to get the SMT solver to generate the necessary
ground instances from $\forall y\, \lnot Q[\vec{\con k}, y]$.
Typically, SMT solvers that reason about quantified formulas use 
heuristic quantifier instantiation techniques based on E-matching~\cite{MouraBjoerner07EfficientEmatchingSmtSolvers}, 
which instantiates universal quantifiers with terms occurring 
in some current set of ground terms built incrementally from the input formula.
%It is fair to say that 
Using E-matching-based heuristic instantiation alone is unlikely 
to be effective in synthesis, where required terms need to be synthesized based 
on the semantics of the input specification.
This is confirmed by our preliminary experiments, even for simple conjectures.
We have developed instead a specialized new technique, which we refer to as \emph{counterexample-guided quantifier instantiation}, 
that allows the SMT solver to quickly converge in many cases to the instantiations 
that refute the negated synthesis conjecture (\ref{eqn:neg_syn_conj_no_syntax}).

The new technique is similar to a popular scheme for synthesis known 
as counterexample-guided inductive synthesis, implemented 
in various synthesis approaches (e.g.,~\cite{SolarLezamaETAL06CombinatorialSketchingFinitePrograms,
JhaETAL10OracleguidedComponentbasedProgramSynthesis}),
but with the major difference of being built-in directly 
into the SMT solver.
The technique is illustrated by the procedure in Figure~\ref{fig:proc1}, 
which grows a set $\Gamma$ of ground instances of $\lnot Q[\vec{\con k}, y]$ 
starting with the formula $\con G \Rightarrow Q[\vec{\con k}, \con e]$
where $\con G$ and $\con e$ are fresh constants of sort $\Bool$ and $\sigma$, 
respectively.
Intuitively, $\con e$ represents a current, partial solution for the original 
synthesis conjecture $\exists f\, \forall \vec x\, Q[\vec x, f(\vec x)]$,
while $\con G$ represents the possibility that the conjecture has 
a (syntactic) solution in the first place.

The procedure, which may not terminate in general, terminates either 
when $\Gamma$ becomes unsatisfiable, in which case it has found a solution, or 
when $\Gamma$ is still satisfiable but all of its models falsify $\con G$, 
in which case the search for a solution was inconclusive.
The procedure is not \define{solution-complete}, that is,
it is not guaranteed to return a solution whenever there is one.
However, thanks to Proposition~\ref{prop:ite-form}, it is \define{solution-sound}:
every $\lambda$-term it returns is indeed a solution 
of the original synthesis problem.

\sparagraph{Finding instantiations.}
The choice of the term $t$ in Step 2 of the procedure is intentionally 
left underspecified because it can be done in a number of ways.
Having a good heuristic for such instantiations is, however, critical 
to the effectiveness of the procedure in practice.
In a $\Sigma$-theory $T$, like integer arithmetic, with a fixed interpretation 
for symbols in $\Sigma$ and a distinguished set of ground $\Sigma$-terms 
denoting the elements of a sort, a simple, if naive, choice for $t$ 
in Figure~\ref{fig:proc1} is the distinguished term denoting the element 
${\con e}^\I$.
For instance, if $\sigma$ is $\Int$ in integer arithmetic, $t$ could be
a concrete integer constant ($0,\pm 1, \pm 2, \ldots$).
This choice amounts to testing whether points in the codomain of the sought 
function $f$ satisfy the original specification $P$.

More sophisticated choices for $t$, in particular where $t$ contains
the variables $\vec x$, may increase the generalization power of this procedure 
and hence its ability to find a solution. 
%For instance, if $T$ is again linear arithmetic and the specification $P[f,x]$ 
%is (trivially) $f(x) = x + 1$, a good heuristics for generating the term $t$ 
%would pick $x + 1$, which leads immediately to the solution $\lambda x\, x+1$.
For instance, our present implementation in the \cvc solver relies 
on the fact that the model $\I$ in Step~2 is constructed from a set 
of equivalence classes over terms computed by the solver during its search.
The procedure selects the term $t$ among those in the equivalence class of $e$, 
other than $e$ itself. 
For instance, consider formula~(\ref{eq:max}) from the previous example 
that encodes the single-invocation form of the specification 
for the max function.
The DPLL(T) architecture, on which \cvc is based, finds a model 
for $Q[\vec{\con a}, \con e ]$ with $\vec{\con a} = (\con a_1, \con a_2)$ 
only if it can first find a subset $M$ of that formula's literals 
that collectively entail $Q[ \vec{\con a}, \con e ]$ at the propositional level.
Due to the last conjunct of~(\ref{eq:max}), $M$ must include either 
$\con e \teq \con a_1$ or $\con e \teq \con a_2$.
Hence, whenever a model can be constructed for $Q[ \vec{\con a}, e ]$,
the equivalence class containing $e$ must contain either $\con a_1$ or $\con a_2$.
Thus using the above selection heuristic, the procedure in Figure~\ref{fig:proc1}
will, after at most two iterations of the loop in Step 2, add the instances 
$\neg Q[ \vec{\con a}, \con a_1 ]$ and $\neg Q[ \vec{\con a}, \con a_2 ]$ 
to $\Gamma$.
As noted in Example~\ref{ex:max}, these two instances are jointly 
$T$-unsatisfiable.
We expect that more sophisticated instantiation techniques can be incorporated.
In particular, both quantifier elimination techniques \cite{Bjoerner10LinearQuantifierEliminationAsAbstractDecision,
Monniaux10QuantifierEliminationLazyModelEnumeration} 
and approaches currently used to infer invariants from templates \cite{MadhavanKuncak14SymbolicResourceBoundInferenceFunctionalPrograms,
Cousot05ProvingProgramInvarianceTerminationParametricAbstraction}
are likely to be beneficial for certain classes of synthesis problems.
The advantage of developing these techniques within an SMT solver is that 
they directly benefit both synthesis and verification in the presence 
of quantified conjectures, thus fostering cross-fertilization between
different fields.

%For instance, a choice for $t$ when $P$ is a formula in quantifier-free linear arithmetic amounts to lazily enumerating the disjuncts arising from a quantifier elimination procedure.

%================================================================================
\section{Refutation-Based Syntax-Guided Synthesis} \label{sec:syntax-guided}
%================================================================================

In syntax-guided synthesis, the functional specification is strengthened 
by an accompanying set of syntactic restrictions on the form of 
the expected solutions.
In a recent line of work~\cite{AlurETAL13SyntaxguidedSynthesis, AlurETAL2014SyGuSMarktoberdorf, DBLP:journals/corr/RaghothamanU14}
these restrictions are expressed
by a grammar $R$ (augmented with a kind of \emph{let} binder) defining 
the language of solution terms, or \define{programs}, for the synthesis problem.
In this section, we present a variant of the approach in the previous section
that incorporates the syntactic restriction directly into the SMT solver 
via a deep embedding of the syntactic restriction $R$ into the solver's logic.
The main idea is to represent $R$ as a set of algebraic datatypes and build 
into the solver an interpretation of these datatypes in terms 
of the original theory $T$.

While our approach is parametric in the background theory $T$ and 
the restriction $R$, it is best explained here with a concrete example. 
%We will use it as a running example for the rest of the section.

\begin{figure}[t]
\centering
$
\begin{array}{l@{\qquad}l}
 \forall x\,y\: \ev( \con x_1, x, y ) \teq x & 
 \forall s_1\, s_2\, x\,y\:
  \ev( \con{leq}(s_1, s_2), x, y ) \teq (\ev( s_1, x, y ) \leq \ev( s_2, x, y )) 
\\[.2ex]
 \forall x\,y\: \ev( \con x_2, x, y ) \teq y &
 \forall s_1\, s_2\, x\,y\:
  \ev( \con{eq}(s_1, s_2), x, y ) \teq (\ev( s_1, x, y ) \teq \ev( s_2, x, y )) 
 \\[.2ex]
 \forall x\,y\: \ev( \con{zero}, x, y ) \teq 0 &
 \forall c_1\, c_2\, x\, y\: 
  \ev( \con{and}(c_1, c_2), x, y ) \teq (\ev( c_1, x, y ) \land \ev( c_2, x, y ))
 \\[.2ex]
 \forall x\,y\: \ev( \con{one}, x, y ) \teq 1 & 
 %\forall c_1 c_2 xy\: 
 % \ev( \con{or}(c_1, c_2), x, y ) \teq (\ev( c_1, x, y ) \lor \ev( c_2, x, y ))
 \forall c\,x\,y\: \ev( \con{not}(c), x, y ) \teq \lnot \ev( c, x, y )
 \\[.2ex]
 \multicolumn{2}{l}{
  \forall s_1\, s_2\, x\,y\: 
   \ev( \con{plus}(s_1, s_2), x, y ) \teq \ev( s_1, x, y ) + \ev( s_2, x, y )}
 \\[.2ex]
 \multicolumn{2}{l}{
  \forall s_1\, s_2\, x\,y\: 
   \ev( \con{minus}(s_1, s_2), x, y ) \teq \ev( s_1, x, y ) - \ev( s_2, x, y )}
 \\[.2ex]
 \multicolumn{2}{l}{
  \forall c\, s_1\, s_2\, x\,y\:
   \ev( \con{if}( c, s_1, s_2 ), x, y ) \teq 
   \ite( \ev( c, x, y ), \ev( s_1, x, y ), \ev( s_2, x, y ) )}
% \\[.2ex]
% \multicolumn{2}{l}{
%  \forall cxy\: \ev( \con{not}(c), x, y ) \teq \lnot \ev( c, x, y )}
\end{array}
$
\caption{Axiomatization of the evaluation operators in grammar $R$ from Example~\ref{ex:max-sygus}.}
\label{fig:ev}
\end{figure}

\begin{example} \label{ex:max-sygus}
Consider again the synthesis conjecture~(\ref{eq:max}) from Example~\ref{ex:max}
but now with a syntactic restriction $R$ for the solution space expressed 
by these %(mutually) 
algebraic datatypes:
\[
 \begin{array}{l@{\quad}l@{\quad}l}
  \con S & := & \con x_1 \mid \con x_2 \mid \con{zero} \mid \con{one} \mid  
                \con{plus}(\con S, \con S) \mid \con{minus}(\con S, \con S) \mid 
                \con{if}( \con C, \con S, \con S ) 
  \\[1ex]
  \con C & := & \con{leq}(\con S, \con S) \mid \con{eq}(\con S, \con S) \mid 
                \con{and}(\con C, \con C) \mid %\con{or}(\con C, \con C) \mid
                \con{not}(\con C)
 \end{array}
\]
The datatypes are meant to encode a term signature that includes nullary
constructors for the variables $x_1$ and $x_2$ of~(\ref{eq:max}), and 
constructors for the symbols of the arithmetic theory $T$.
Terms of sort $\con S$ (resp., $\con C$) refer to theory terms of sort $\Int$ 
(resp., $\Bool$).
%\footnote{
%It is important to note that the symbols shown in the definition of $S$ and $C$ denote datatype constructors, 
%and are not to be confused with the builtin theory operators they correspond to.
%This will be unambiguous from the context in which we use these symbols.
%}

Instead of the theory of linear integer arithmetic, we now consider 
its combination $\TD$ with the theory of the datatypes above extended 
with two \define{evaluation operators}, that is, two function symbols 
$\ev^{\con S \times \Int \times \Int \to \Int}$ and 
$\ev^{\con C \times \Int \times \Int \to \Bool}$ respectively embedding 
$\con S$ in $\Int$ and $\con C$ in $\Bool$.
We define $\TD$ so that all of its models satisfy the formulas in Figure~\ref{fig:ev}.
% We also omit the admission of omitting...
%\footnote{
%We omit a formal definition of $\TD$ for space constraints,
%hoping that the given description provides enough of an intuition.
%}
The evaluation operators effectively define an interpreter for programs 
(i.e., terms of sort $\con S$ and $\con C$) with input parameters $x_1$ and $x_2$.

It is possible to instrument an SMT solver that support user-defined datatypes,
quantifiers and linear arithmetic so that it constructs automatically 
from the syntactic restriction $R$ both the datatypes $\con S$ and $\con C$ and 
the two evaluation operators.
Reasoning about $\con S$ and $\con C$ is done by the built-in subsolver 
for datatypes.
Reasoning about the evaluation operators is achieved by reducing ground terms 
of the form $\ev(d, t_1, t_2)$ to smaller terms by means 
of selected instantiations of the axioms from Figure~\ref{fig:ev},
with a number of instances proportional to the size of term $d$.
It is also possible to show that $\TD$ is satisfaction complete 
with respect to the class
\begin{eqnarray*}
\lan_2 & := & \{
 \exists g\, \forall \vec z\, P[\lambda \vec z.\, \ev(g, \vec z),\, \vec x] 
 \mid
 P[f, \vec x] \in \props
\}
\end{eqnarray*}
where instead of terms of the form $f(t_1, t_2)$ in $P$ we have, 
modulo $\beta$-reductions, terms of the form $\ev(g, t_1, t_2)$.\footnote{
We stress again, that both the instrumentation of the solver and 
the satisfaction completeness argument for the extended theory are generic 
with respect to the syntactic restriction on the synthesis problem and 
the original satisfaction complete theory $T$.
}
For instance, the formula $P[f, \vec x]$ in Equation~(\ref{eq:max-orig}) 
from Example~\ref{ex:max} can be restated in $\TD$ as the formula below
where $g$ is a variable of type $\con S$:
\begin{eqnarray*}
 P_\ev[ g, \vec x ] & := & 
 \ev( g, \vec x ) \geq x_1 \land \ev( g,\vec x ) \geq x_2 \land
 ( \ev( g, \vec x ) \teq x_1 \lor \ev( g,\vec x ) \teq x_2 )
\end{eqnarray*} 
%
%(and $\ev$ is the one with type $\con S \times \Int \times \Int \to \Int$).
In contrast to $P[f, \vec x]$, the new formula $P_\ev[ g, \vec x ]$ is 
first-order, with the role of the second-order variable $f$ now played 
by the first-order variable $g$.
% which encodes the accepted syntax for solutions for $f$.

When asked for a solution for~(\ref{eq:max-orig}) under the restriction $R$,
the instrumented SMT solver will try to determine instead the 
$\TD$-unsatisfiability of $\forall g\, \exists \vec x\, \lnot P_\ev[g, \vec x]$.
Instantiating $g$ in the latter formula with 
$s := \con{if}( \con{leq}(\con x_1, \con x_2), \con x_2, \con x_1 )$, say, produces 
a formula that the solver can prove to be $\TD$-unsatisfiable.
This suffices to show that the program $\ite(x_1 \leq x_2, x_2, x_1)$, 
the analogue of $s$ in the language of $T$, is a solution 
of the synthesis conjecture~(\ref{eq:max-orig}) under the syntactic restriction $R$.
\qed
\end{example}

\begin{figure}[t]
\begin{enumerate}
\item
$\Gamma := \emptyset$
\item Repeat
 \begin{enumerate}
  \item \label{it:model-i}
  Let $\vec{\con k}$ be a tuple of distinct fresh constants. \\
   If there is a model $\I$ of $\TD$ satisfying $\Gamma$ \emph{and} $\con G$, then
   $\Gamma := \Gamma \cup \{ \lnot P_\ev[{\con e}^\I, \vec{\con k}] \}$ ; \\
   otherwise, return ``no solution found''
  \item \label{it:model-j}
   If there is a model $\mathcal J$ of $\TD$ satisfying $\Gamma$, then
   $\Gamma := \Gamma \cup \{ \con{G} \Rightarrow P_\ev[\con e, \vec{\con k}^{\mathcal J}] \}$ ; \\
   otherwise, return ${\con e}^\I$ as a solution
 \end{enumerate}
\end{enumerate}
\vspace{-2ex}
\caption{A refutation-based syntax-guided synthesis procedure for $\exists f\, \forall \vec x\, P_\ev[f,\vec x]$.
}
\label{fig:proc2}
\end{figure}

To prove the unsatisfiability of formulas like
$\forall g\, \exists \vec x\, \lnot P_\ev[g, \vec x]$
in the example above we use a procedure similar to that 
in Section~\ref{sec:refutation-based}, 
but specialized to the extended theory $\TD$.
The procedure is described in Figure~\ref{fig:proc2}.
Like the one in Figure~\ref{fig:proc1}, it uses an uninterpreted constant 
$\con e$ representing a solution candidate, and a Boolean variable $\con G$ 
representing the existence of a solution. 
The main difference, of course, is that now $\con e$ ranges over the datatype
representing the restricted solution space.
In any model of $\TD$, a term of datatype sort evaluates to a term built 
exclusively with constructor symbols.
This is why the procedure returns in Step~\ref{it:model-j} the value of $\con e$ 
in the model $\I$ found in Step~\ref{it:model-i}.
As we showed in the previous example, a program that solves the original problem 
can then be reconstructed from the returned datatype term.

\begin{figure}[t]
\[\begin{array}{c@{\hspace{1em}}l@{\hspace{1em}}l}
\hline
\text{Step} & \text{Model} & \text{Added Formula} \ 
\\
\hline
 \ref{it:model-i} & \{ \con e \mapsto \con x_1, \ldots \} &
 \lnot P_\ev[ \con x_1, \con a_1, \con b_1 ]
 \\
 \ref{it:model-j} & \{ \con a_1 \mapsto 0, \con b_1 \mapsto 1, \ldots \} & 
 \con G \Rightarrow P_\ev[ \con e, 0, 1 ] 
 \\
 \ref{it:model-i} & \{ \con e \mapsto \con x_2, \ldots \} & 
 \lnot P_\ev[ \con x_2, \con a_2, \con b_2 ] 
 \\
 \ref{it:model-j} & \{ \con a_2 \mapsto 1, \con b_2 \mapsto 0, \ldots \} &
  \con G \Rightarrow P_\ev[ \con e, 1, 0 ] 
  \\
 \ref{it:model-i} & \{ \con e \mapsto \con{one}, \ldots \} &
 \lnot P_\ev[ \con{one}, \con a_3, \con b_3 ] 
 \\
 \ref{it:model-j} & \{ \con a_3 \mapsto 2, \con b_3 \mapsto 0, \ldots \} & 
 \con G \Rightarrow P_\ev[ \con e, 2, 0 ] 
 \\
 \ref{it:model-i} & \{ \con e \mapsto \con{plus}(\con x_1, \con x_2), \ldots \} &
 \lnot P_\ev[ \con{plus}(\con x_1, \con x_2), \con a_4, \con b_4 ] 
 \\
 \ref{it:model-j} & \{ \con a_4 \mapsto 1, \con b_4 \mapsto 1, \ldots \} & 
 \con G \Rightarrow P_\ev[ \con e, 1, 1 ] 
 \\
 \ref{it:model-i} & \{ \con e \mapsto \con{if}( \con{leq}(\con x_1, \con{one}), \con{one}, \con x_1 ), \ldots \} & 
 \lnot P_\ev[ \con{if}( \con{leq}(\con x_1, \con{one}), \con{one}, \con x_1 ), \con a_5, \con b_5 ] 
 \\
 \ref{it:model-j} & \{ \con a_5 \mapsto 1, \con b_5 \mapsto 2, \ldots \} & 
 \con G \Rightarrow P_\ev[ \con e, 1, 2 ] 
 \\
 \ref{it:model-i} & \{ \con e \mapsto \con{if}( \con{leq}(\con x_1, \con x_2), \con x_2, \con x_1 ), \ldots \} & 
 \lnot P_\ev[ \con{if}( \con{leq}(\con x_1, \con x_2), \con x_2, \con x_1 ), \con a_6, \con b_6 ] 
 \\
 \ref{it:model-j} & \text{none} &  \\
\hline
\end{array}
\]
\smallskip

For $i=1,\ldots,6$, 
$\con a_i$ and $\con b_i$ are fresh constants of type $\Int$.

\caption{A run of the procedure from Figure~\ref{fig:proc2}.}
\label{fig:run}
\end{figure}

\sparagraph{Implementation.}
We implemented the procedure in the \cvc solver. %~\cite{CVC4-CAV-11}.
Figure~\ref{fig:run} shows a run of that implementation over the conjecture 
from Example~\ref{ex:max-sygus}.
In this run, note that each model found for $\con e$ satisfies all values 
of counterexamples found for previous candidates.
After the sixth iteration of Step~\ref{it:model-i}, the procedure finds 
the candidate $\con{if}( \con{leq}(\con x_1, \con x_2), \con x_2, \con x_1 )$, 
for which no counterexample exists, indicating that the procedure has found 
a solution for the synthesis conjecture.
Currently, this problem can be solved in about $0.5$ seconds 
in the latest development version of \cvc.

To make the procedure practical it is necessary to look 
for \emph{small} solutions to synthesis conjectures.
A simple way to limit the size of the candidate solutions is to consider 
smaller programs before larger ones.
Adapting techniques for finding finite models of minimal size~\cite{reynolds2013finite}, 
we use a strategy that starting, from $n = 0$, searches for programs 
of size $n+1$ only after its has exhausted the search for programs of size $n$.
In solvers based on the DPLL($T$) architecture, like \cvc,
this can be accomplished by introducing a splitting lemma of the form 
$( \size( \con e ) \leq 0 \lor \lnot \size( \con e ) \leq 0 )$ and
asserting $\size( \con e ) \leq 0$ as the first decision literal, 
where $\size$ is a function symbol of type $\sigma \to \Int$ 
for every datatype sort $\sigma$ and stands for the function
that maps each datatype value to its term size 
(i.e., the number of non-nullary constructor applications in the term).
We do the same for $\size( \con e ) \leq 1$ if and when $\lnot \size( \con e ) \leq 0$ becomes asserted.
We extended the procedure for algebraic datatypes in \cvc~\cite{BarST-JSAT-07} to handle constraints involving $\size$.
The extended procedure remains a decision procedure for input problems
with a concrete upper bound on terms of the form $\size(u)$, 
for each variable or uninterpreted constant $u$ of datatype sort in the problem.
This is enough for our purposes since the only term $u$ like that
in our synthesis procedure is $\con e$.

\begin{proposition} \em
\label{prop:sygus-sound-complete}
With the search strategy above, the procedure in Figure~\ref{fig:proc2} has 
the following properties:
\begin{enumerate}
\item (Solution Soundness) 
Every term it returns can be mapped to a solution of the original synthesis conjecture 
$\exists f\,\forall \vec x\, P[f, \vec x]$ under the restriction $R$.
\item (Refutation Soundness) 
If it answers ``no solution found'', the original conjecture
has no solutions under the restriction $R$.
\item (Solution Completeness) 
If the original conjecture has a solution under $R$,
the procedure will find one.
\end{enumerate}
\end{proposition}
\begin{proof}
To show solution soundness,
consider the case when the procedure returns $e^\I$ as a solution.
Then, $\Gamma \cup \neg P_\ev[ e^\I, \vec{\con k}]$ is $\TD$-unsatisfiable
for some $\Gamma, \vec{\con k}$, where $\Gamma$ is $\TD$-satisfiable and $\vec{\con k}$ is a tuple of distinct fresh constants.
Since $\vec{\con k}$ are fresh, $\Gamma \cup \exists \vec{x}\, \neg P_\ev[ e^\I, \vec{x}]$ is $\TD$-unsatisfiable.
Since $\Gamma$ is $\TD$-satisfiable and $\Gamma \cup \exists \vec{x}\, \neg P_\ev[ e^\I, \vec{x}]$ is not, 
then at least one model of $\TD$ (namely, one for $\Gamma$) does not satisfy $\exists \vec{x}\, \neg P_\ev[ e^\I, \vec{x}]$.
Thus, since $\TD$ is satisfaction complete, no models of $\TD$ satisfy $\exists \vec{x}\, \neg P_\ev[ e^\I, \vec{x}]$,
and thus all models of $\TD$ satisfy $\forall \vec{x}\, P_\ev[ e^\I, \vec{x}]$.
Assuming our translation from $P$ to $P_\ev$ is faithful, 
the analogue of $e^\I$ in the language of $T$ is a solution for the conjecture $\exists f\,\forall \vec x\, P[f, \vec x]$.

To show refutation soundness, 
consider the case when the procedure returns ``no solution found".
Then, there exists a $\Gamma = ( \Gamma' \cup G \Rightarrow P_\ev[ e, \vec{\con k}^{\mathcal J}] )$ such that
$\Gamma'$ is $\TD$-satisfiable, and $\Gamma \cup G$ is $\TD$-unsatisfiable.
Clearly based on the clauses added by the procedure, we have that $\Gamma$ is equivalent to 
$\Gamma'' \cup G \Rightarrow ( P_\ev[ \con e, \vec{\con u}_1] \wedge \ldots \wedge P_\ev[ \con e, \vec{\con u}_n] )$,
for some $\vec{\con u}_1 \ldots \vec{\con u}_n$ where $\Gamma'' \subseteq \Gamma'$ is $\TD$-satisfiable and does not contain $G$ or $\con e$.
Since $\Gamma \cup G$ is $\TD$-unsatisfiable, we have that 
$\Gamma'' \cup P_\ev[ \con e, \vec{\con u}_1] \wedge \ldots \wedge P_\ev[ \con e, \vec{\con u}_n]$ is $\TD$-unsatisfiable.
Since $\Gamma''$ does not contain $\con e$, 
$\Gamma'' \cup \exists y\, ( P_\ev[ y, \vec{\con u}_1] \wedge \ldots \wedge P_\ev[ y, \vec{\con u}_n] )$ is $\TD$-unsatisfiable.
Since $\TD$ is satisfaction complete and $\Gamma''$ is $\TD$-satisfiable, 
$\exists y\, ( P_\ev[ y, \vec{\con u}_1] \wedge \ldots \wedge P_\ev[ y, \vec{\con u}_n] )$ is $\TD$-unsatisfiable.
Thus, $\exists y\, ( P_\ev[ y, \vec{\con u}_1] \wedge \ldots \wedge P_\ev[ y, \vec{\con u}_n] )$ is $\TD$-unsatisfiable,
and thus $\exists y\, \forall \vec{x}\, P_\ev[ y, \vec{x}]$ is $\TD$-unsatisfiable.
Assuming our translation from $P$ to $P_\ev$ is faithful, 
this implies there is no solution for the conjecture $\exists f\,\forall \vec x\, P[f, \vec x]$.

Given solution and refutation soundness of the procedure, 
to show the procedure is solution complete, 
it suffices to show that the procedure terminates when the original conjecture has a solution under $R$.
Let $\lambda \vec x.\ t$ be such a solution, and let $d$ be the analogue of $t$ in the language of $\TD$.
Let $n$ be equal to the number of datatypes of the same type as $d$ that are at most the size of $d$, 
which we know is finite.
For $i = 1, 2, \ldots$, let $\I_i$ and $\mathcal J_i$ be the models found on the $i^{th}$ iteration of Steps~\ref{it:model-i} and~\ref{it:model-j} respectively.
%We know that $\con e^{\I_k} \neq \con e^{\I_j}$ for $k > j$.
%To show this, 
Assume the procedure runs at least $k$ iterations, and let $1 \leq j < k$.
Since $\mathcal J_j$ satisfies $\neg P_\ev[e^{\I_j}, \vec{\con k}]$,
all models of $\TD$ satisfy $\neg P_\ev[\con e^{\I_j}, \vec{\con k}^{\mathcal J_j}]$
since $\TD$ is satisfaction complete.
Since $\I_k$ satisfies $G$, it must also satisfy $P_\ev[\con e, \vec{\con k}^{\mathcal J_j}]$, and thus $\con e^{\I_k} \neq \con e^{\I_j}$.
Thus, each $e^{\I_1}, e^{\I_2}, \ldots$ is distinct, and
the procedure in Figure~\ref{fig:proc2} executes at most $n$ iterations of Step~\ref{it:model-i}.
Since the background theory $\TD$ is decidable, Steps~\ref{it:model-i} and~\ref{it:model-j} are terminating,
and thus the procedure is terminating when a solution exists.
\qed
\end{proof}

Note that by this proposition the procedure can diverge only 
if the input synthesis conjecture has no solution.
%We refer the reader to a longer version of this paper
%for a proof of Proposition~\ref{prop:sygus-sound-complete}~\cite{ReyEtAl-RR-2015}.
\begin{comment}
For a general idea, the proof of solution soundness is based 
on the observation that when the procedure terminates at Step~\ref{it:model-j},
$\Gamma$ has an unsatisfiable core with just one instance of $\lnot P[g, \vec x]$.
The procedure is refutation sound since when no model of $\Gamma$ 
in Step~\ref{it:model-i} satisfies $\con G$, 
we have that even an arbitrary $\con e$ cannot satisfy the current set 
of instances added to $\Gamma$ in Step~\ref{it:model-j}. 
Finally, the procedure is solution complete first of all because
Step~\ref{it:model-i} and~\ref{it:model-j} are effective thanks to 
the decidability of the background theory $\TD$.
Each execution of Step~\ref{it:model-i} is guaranteed to produce 
a new candidate since $\TD$ is also satisfaction complete.
Thus, in the worst case, the procedure amounts an enumeration
of all possible programs until a solution is found.
\end{comment}

%--------------------------------------------------------------------------------
\section{Single Invocation Techniques for Syntax-Guided Problems}
\label{sec:si-syntax-guided}
%--------------------------------------------------------------------------------

In this section, we considered the combined case of 
\emph{single-invocation synthesis conjectures with syntactic restrictions}.
Given a set $R$ of syntactic restrictions expressed 
by a datatype $\con S$ for programs and a datatype $\con C$ 
for Boolean expressions, consider the case where 
$(i)$ $\con S$ contains the constructor 
$\mathsf{if} : \con C \times \con S \times \con S \rightarrow \con S$ 
(with the expected meaning) and 
$(ii)$ the function to be synthesized is specified by a single-invocation 
property that can be expressed as a term of sort $\con C$.
This is the case for the conjecture from Example~\ref{ex:max-sygus} where
the property $P_\ev[g, \vec x]$ can be rephrased as:
\begin{eqnarray} \label{eqn:gen-sygus}
P_{\con C}[g, \vec x] & := & \ev( \con{and}( \con{leq}( \con x_1, g ), \con{and}( \con{leq}( \con x_2, g ), \con{or}( \con{eq}( g, \con x_1 ), \con{eq}( g, \con x_2 ) ) ) ), \vec x)
\end{eqnarray}
where again $g$ has type $\con S$, $\vec x = (x_1, x_2)$, and $x_1$ and $x_2$ 
have type $\Int$.
The procedure in Figure~\ref{fig:proc1} can be readily modified to apply 
to this formula, with $P_{\con C}[g, \vec{\con k}]$ and $g$ taking the role 
respectively of $Q[\vec{\con k}, y]$ and $y$ in that figure, since it generates
solutions meeting our syntactic requirements.
Running this modified procedure instead the one in Figure~\ref{fig:proc2}
has the advantage that only the outputs of a solution need 
to be synthesized, not conditions in $\mathsf{ite}$-terms.
%and moreover maintains the benefits of the theoretical properties as stated in Proposition~\ref{prop:sygus-sound-complete}.
However, in our experimental evaluation found that the overhead of using 
an embedding into datatypes % and evaluation operators 
for syntax-guided problems is significant with respect to the performance 
of the solver on problems with no syntactic restrictions.
For this reason, we advocate an approach for single-invocation synthesis 
conjectures with syntactic restrictions that runs the procedure 
from Figure~\ref{fig:proc1} as is, ignoring the syntactic restrictions $R$,
and subsequently reconstructs from its returned solution one satisfying 
the restrictions.
For that it is useful to assume that terms $t$ in $T$ can be effectively reduced 
to some ($T$-equivalent and unique) \emph{normal form}, 
which we denote by $\nmf t$.

Say the procedure from Figure~\ref{fig:proc1} returns 
a solution $\lambda \vec x.\, t$ for a function $f$.
To construct from that a solution that meets 
the syntactic restrictions specified by datatype $\con S$,
we run the iterative procedure described in Figure~\ref{fig:proc3}.
This procedure maintains an evolving set $A$ of triples of the form $( t, s, D )$,
where $D$ is a datatype, $t$ is a term in normal form, $s$ is a term
satisfying the restrictions specified by $D$.
The procedure incrementally makes calls to the subprocedure \rcon, 
which takes a normal form term $t$, a datatype $D$ and the set $A$ above, and
returns a pair $( s, U )$ where
$s$ is a term equivalent to $t$ in $T$, % representing a candidate solution,
and $U$ is a set of pairs $(s', D')$ where 
$s'$ is a subterm of $s$ that fails to satisfy the syntactic restriction expressed by datatype $D'$.
Overall, the procedure alternates between calling \rcon\ and 
adding triples to $A$ until $\rcon( t, D, A )$ returns a pair of the form 
$( s, \emptyset )$, in which case $s$ is a solution satisfying 
the syntactic restrictions specified by $\con S$.
%We demonstrate this process with an example.

\begin{figure}[t]
\begin{enumerate}
\item $A : = \emptyset$ ; $t' := \nmf{t}$
\item for $i = 1, 2, \ldots$
  \begin{enumerate}  
  \item $( s, U ) := \rcon( t', \con S, A )$; 
  \item \label{it:enum}
  if $U$ is empty, return $s$;
  otherwise, for each datatype $D_j$ occurring in $U$
  \begin{itemize}
  \item[] let $d_i$ be the $i^{th}$ term in a fair enumeration of the elements of $D_j$
  \item[] let $t_i$ be the analogue of $d_i$ in the background theory $T$
  \item[] add $( \nmf{t_i}, t_i, D_j )$ to $A$
  \end{itemize}
  \end{enumerate}
%\end{framed}
\end{enumerate}

\ \rcon$( t, D, A )$
\vspace{-2ex}
\begin{itemize}
 \item[\ ]
 if $(t, s, D) \in A$, return $( s, \emptyset )$;
 otherwise, do one of the following: \\
 \begin{tabular}{ll}
  (1) & choose a $f( t_1, \ldots, t_n )$ s.t. $\nmf{ f( t_1, \ldots, t_n ) }\ = t$ and
        $f$ has an analogue $c^{D_1 \ldots D_n D}$ in $D$ \\
      & let $( s_i, U_i ) = \rcon( \nmf{t_i}, D_i, A )$ for $i = 1, \ldots, n$ \\
      & return $( f( s_1, \ldots, s_n ), U_1 \cup \ldots \cup U_n )$ \\
  (2) & return $( t, \{ ( t, D ) \} )$
  \end{tabular}
\end{itemize}
\vspace{-2ex}
\caption{A procedure for finding a term equivalent to $t$ that meets the syntactic restrictions specified by datatype $\con S$.  
%The subprocedure \rcon returns a pair $( s, U )$, where $U$ is
%a set of pairs of the form $( t_i, S_i )$ representing that it remains to find a term equivalent to $t_i$ that meets syntactic restriction specified by datatype $S_i$.
%When the set of pairs $U$ is empty, the procedure is successful, and returns the term $s$.
%We write $\nmf{t_i}$ to denote the normal form of term $t_i$.
}
\label{fig:proc3}
\end{figure}

\begin{example}
Say we wish to construct a solution equivalent to $\lambda x_1\,x_2.\: x_1+(2*x_2)$ that meets restrictions specified by datatype $\con S$ from Example~\ref{ex:max-sygus}.
To do so, we let $A = \emptyset$, and call $\rcon(\nmf{(x_1+(2*x_2))}, \con S, A )$.
Since $A$ is empty and $+$ is the analogue of constructor $\con{plus}^{ \con S \con S \con S}$ of $\con S$, assuming $\nmf{(x_1+(2*x_2))}\ = x_1+(2*x_2)$,
we may choose to return a pair based on the result of calling $\rcon$ on $\nmf{x_1}$ and $\nmf{(2*x_2)}$.
Since $\con x_1^{\con S}$ is a constructor of $\con S$ and $\nmf{x_1}\ = x_1$,
$\rcon( x_1, \con S, A )$ returns $( x_1, \emptyset )$.
Since $\con S$ does not have a constructor for $*$, 
we must either choose a term $t$ such that $\nmf{t}\ = \nmf{(2*x_2)}$ where the topmost symbol of $t$ is the analogue of a constructor in $\con S$, 
or otherwise return the pair $( 2*x_2, \{ (2*x_2, \con S ) \} )$.
Suppose we do the latter, 
%and thus $\rcon( x_1+2*x_2, \con S, A )$ returns $( 2*x_2, \{ (2*x_2, \con S ) \} )$,
and thus $\rcon( x_1+(2*x_2), \con S, A )$ returns $( x_1+(2*x_2), \{ (2*x_2, \con S ) \} )$.
Since the second component of this pair is not empty,
we pick in Step~\ref{it:enum} the first element of $\con S$, $\con x_1$ say, and add $( x_1, x_1, \con S )$ to $A$.
%This indicates that there is a term in grammar $\con S$ (as witnessed by $\con x_1$) that is equivalent to $x_1$.
We then call $\rcon( \nmf{(x_1+(2*x_2))}, \con S, A )$ which by the same strategy above returns $( x_1+(2*x_2), \{ (2*x_2, \con S ) \} )$.
This process continues until we pick, the term $\con{plus}( \con{x_2}, \con{x_2} )$ say,
whose analogue is $x_2+x_2$.
Assuming $\nmf{(x_2+x_2)}\ = \nmf{(2*x_2)}$,
after adding the pair $( 2*x_2, x_2+x_2, \con S )$ to $A$,
$\rcon( \nmf{(x_1+(2*x_2))}, \con S, A )$ returns the pair $( x_1+(x_2+x_2), \emptyset )$,
indicating that $\lambda x_1\,x_2.\, x_1+(x_2+x_2)$ is equivalent to $\lambda x_1\,x_2.\, x_1+(2*x_2)$, and meets the restrictions specified by $\con S$.
\qed
\end{example}

This procedure depends upon the use of normal forms for terms.
It should be noted that, since the top symbol of $t$ is generally $\ite$,
this normalization includes both low-level rewriting of literals within $t$,
but also includes high-level rewriting techniques such as $\ite$ simplification, redundant subterm elimination and destructive equality resolution.
Also, notice that we are not assuming that $\nmf{t}\ = \nmf{s}$ if and only if $t$ is equivalent to $s$,
and thus normal forms only underapproximate an equivalence relation between terms.
%and thus checking whether the normal forms of two terms are identical underapproximates whether they are indeed equivalent.
Having a (more) consistent normal form for terms 
allows us to compute a (tighter) underapproximation, thus improving the performance of the reconstruction.
In this procedure, we use the same normal form for terms that is used by the individual decision procedures of \cvc.
This is unproblematic for theories such as linear arithmetic whose normal form for terms is a sorted list of monomials,
but it can be problematic for theories such as bitvectors.
As a consequence, we use several optimizations, 
omitted in the description of the procedure in Figure~\ref{fig:proc3}, 
to increase  the likelihood that the procedure terminates 
in a reasonable amount of time.
For instance, in our implementation
the return value of $\rcon$ is not recomputed every time $A$ is updated.
Instead, we maintain an evolving directed acyclic graph (dag),
whose nodes are pairs $( t, S )$ for term $t$ and datatype $S$ (the terms we have yet to reconstruct), 
and whose edges are the direct subchildren of that term.
Datatype terms are enumerated for all datatypes in this dag,
which is incrementally pruned as pairs are added to $A$ until it becomes empty.
Another optimization is that the procedure \rcon\ may choose to try simultaneously to reconstruct
\emph{multiple} terms of the form $f( t_1, \ldots, t_n )$ when matching a term $t$ to a syntactic specification $S$, 
reconstructing $t$ when any such term can be reconstructed.

Although the overhead of this procedure can be significant
when large subterms do not meet the syntactic restrictions, 
we found that in practice it quickly terminates successfully 
for a majority of the solutions we considered where reconstruction was possible,
as we discuss in the next section.
Furthermore, it makes our implementation more robust,
since it effectively treats in the same way different properties that are equal modulo normalization (which is parametric in the built-in theories we consider).

\section{Experimental Evaluation}

We implemented the techniques from the previous sections in the SMT solver \cvc~\cite{CVC4-CAV-11},
which has support for quantified formulas and a wide range of theories including arithmetic, bitvectors, and algebraic datatypes.
We evaluated our implementation on 243 benchmarks used in the SyGuS 2014 competition~\cite{AlurETAL2014SyGuSMarktoberdorf} that were publicly available on the StarExec execution service~\cite{StuST-IJCAR-14}.
The benchmarks are in a new format for specifying syntax-guided synthesis problems~\cite{DBLP:journals/corr/RaghothamanU14}.
We added parsing support to \cvc for most features of this format.
All SyGuS benchmarks considered contain synthesis conjectures whose background theory is either linear integer arithmetic or bitvectors.
We made some minor modifications to benchmarks to avoid naming conflicts,
and to explicitly define several bitvector operators that are not supported natively by \cvc.

We considered multiple configurations of \cvc corresponding to the techniques mentioned in this paper.
Configuration {\bf cvc4+sg} executes the syntax-guided procedure from Section~\ref{sec:syntax-guided},
even in cases where the synthesis conjecture is single-invocation.
Configuration {\bf cvc4+si-r} executes the procedure from Section~\ref{sec:refutation-based} 
on all benchmarks having conjectures that it can deduce are single-invocation.
In total, it discovered that 176 of the 243 benchmarks could be rewritten into a form that was single-invocation.
This configuration simply ignores any syntax restrictions 
on the expected solution.
Finally, configuration {\bf cvc4+si} uses the same procedure used 
by {\bf cvc4+si-r} but then attempts to reconstruct any found solution
as a term in required syntax, as described in Section~\ref{sec:si-syntax-guided}.

We ran all configurations on all benchmarks on the StarExec cluster.\footnote{
%The results can be found at 
%{\scriptsize \url{https://www.starexec.org/starexec/secure/details/job.jsp?id=6561}} 
%and {\scriptsize \url{https://www.starexec.org/starexec/secure/details/job.jsp?id=6440}}.
A detailed summary can be found at {\scriptsize \url{http://lara.epfl.ch/w/cvc4-synthesis}.}}
%http://lara.epfl.ch/~reynolds/CAV2015-synth}.}}
We provide comparative results here primarily against the enumerative CEGIS solver \esolver~\cite{Udupa2013}, the winner of the SyGuS 2014 competition.
In our tests, we found that \esolver performed significantly better than the other entrants of that competition.

\begin{figure}[t]
\centering
{\scriptsize
\begin{tabular}{|l|cr|cr|cr|cr|cr|cr|cr|cr|cr|}
\hline                                                                
  & \multicolumn{2}{c|}{{\bf array} (32)}     & \multicolumn{2}{c|}{{\bf bv} (7)}     & \multicolumn{2}{c|}{{\bf hd} (56)}      & \multicolumn{2}{c|}{{\bf icfp} (50)}      & \multicolumn{2}{c|}{{\bf int} (15)}     & \multicolumn{2}{c|}{{\bf let} (8)}      & \multicolumn{2}{c|}{{\bf multf} (8)}      & \multicolumn{2}{c|}{{\bf Total} (176)}    
\\                                                                
  & \#  & time  & \#  & time  & \#  & time  & \#  & time  & \#  & time  & \#  & time  & \#  & time  & \#  & time
\\                                                                
\hline                                                                
{\bf esolver} & 4 & 2250.7  & 2 & 71.2  & 50  & 878.5 & 0 & 0 & 5 & 1416.7  & 2 & 0.0 & 7 & 0.6 & 70  & 4617.7
\\                                                                
{\bf cvc4+sg} & 1 & 3.1 & 0 & 0 & 34  & 4308.9  & 1 & 0.5 & 3 & 1.7 & 2 & 0.5 & 7 & 628.3 & 48  & 4943
\\                                                                
{\bf cvc4+si-r} & (32)  & 1.2 & (6) & 4.7 & (56)  & 2.1 & (43)  & 3403.5  & (15)  & 0.6 & (8) & 1.0 & (8) & 0.2 & (168) & 3413.3
\\                                                                
{\bf cvc4+si} & 30  & 1449.5  & 5 & 0.1 & 52  & 2322.9  & 0 & 0 & 6 & 0.1 & 2 & 0.5 & 7 & 0.1 & 102 & 3773.2
\\                                                                
\hline                                                                
\end{tabular}
\\
}
\caption{Results for single-invocation synthesis conjectures, showing times 
(in seconds) and number of benchmarks solved
by each solver and configuration over 8 benchmark classes with a 3600s timeout. 
The number of benchmarks solved by configuration {\bf cvc4+si-r} are in parentheses
because its solutions do not necessarily satisfy the given syntactic restrictions.}
\label{fig:results-solved-si}
\end{figure}

\sparagraph{Benchmarks with single-invocation synthesis conjectures.}
The results for benchmarks with single-invocation properties are shown
in Figure~\ref{fig:results-solved-si}.
Configuration {\bf cvc4+si-r} found a solution (although not necessarily in 
the required language) very quickly for a majority of benchmarks.
It terminated successfully for 168 of 176 benchmarks, and in less than a second
for 159 of those.
%\footnote{It is important to note that this number is a coarse overestimation of the number of solutions that can be reconstructed.
%For instance, the {\bf icfp} contain benchmarks whose conjecture lists a set of input/output pairs that a function must satisify.
%However, the grammar for solutions does not directly include an $ite$ construct, meaning that generalization is required.
%For handling such benchmarks, the method in Figure~\ref{fig:proc1} additionally requires as a termination condition that the unsatisfiable core contains only one instance.
%We do not explore this restriction in this paper.
%}
Not all solutions found using this method met the syntactic restrictions.
Nevertheless, our methods for reconstructing these solutions into the required grammar, implemented in configuration {\bf cvc4+si}, succeeded in 102 cases,
or 61\% of the total.
This is 32 more benchmarks than the 70 solved by \esolver, 
the best known solver for these benchmarks so far.
In total, {\bf cvc4+si} solved 34 benchmarks that \esolver did not, while \esolver solved 2 that {\bf cvc4+si} did not.

The solutions returned by {\bf cvc4+si-r} were often large, having 
an order of 10K subterms for harder benchmarks.
However, after exhaustively applying simplification techniques during reconstruction with configuration {\bf cvc4+si},
we found that the size of those solutions is comparable to other solvers, 
and in some cases even smaller.
For instance, among the 68 benchmarks solved by both \esolver and {\bf cvc4+si},
the former produced a smaller solution in 15 cases and the latter in 9.
Only in 2 cases did {\bf cvc4+si} produce a solution that had 10 more subterms
than the solution produced by \esolver.
This indicates that in addition to having a high precision,
the techniques from Section~\ref{sec:si-syntax-guided} used 
for solution reconstruction are effective also at producing succinct solutions 
for this benchmark library.

Configuration {\bf cvc4+sg} does not take advantage of the fact
that a synthesis conjecture is single-invocation.
However, it was able to solve 48 of these benchmarks, including a small number
not solved by any other configuration, like one
from the {\bf icfp} class whose solution was a single argument function 
over bitvectors that shifted its input right by four bits.
In addition to being solution complete,
{\bf cvc4+sg} always produces solutions of minimal term size,
something not guaranteed by the other solvers and \cvc configurations.
Of the 47 benchmarks solved by both {\bf cvc4+sg} and \esolver,
the solution returned by {\bf cvc4+sg} was smaller than 
the one returned by \esolver in 6 cases, and had the same size in the others.
This provides an experimental confirmation that the fairness techniques 
for term size described in Section~\ref{sec:syntax-guided} ensure minimal size solutions.

\begin{figure}[t]
\centering
{\scriptsize
\begin{tabular}{|l|cc|cc|cc|cc|cc|}                                 
\hline                                        
  & \multicolumn{2}{c|}{{\bf int} (3)}      & \multicolumn{2}{c|}{{\bf invgu} (28)}     & \multicolumn{2}{c|}{{\bf invg} (28)}      & \multicolumn{2}{c|}{{\bf vctrl} (8)}      & \multicolumn{2}{c|}{{\bf Total} (67)}   
\\                                        
  & \#  & time  & \#  & time  & \#  & time  & \#  & time  & \#  & time
\\                                                                            
\hline                                        
{\bf esolver} & 3 & 1.6 & 25  & 86.3  & 25  & 85.6  & 5 & 29.5  & 58  & 203.0
\\                                        
{\bf cvc4+sg} & 3 & 1476.0  & 23  & 811.6 & 22  & 2283.2  & 5 & 2933.1  & 53  & 7503.9
\\                                        
\hline                                                                                                
\end{tabular}
\\
}
\caption{Results for synthesis conjectures that are not single-invocation, showing times (in seconds) and numbers of benchmarks solved
by \cvc and \esolver over 4 benchmark classes with a 3600s timeout. }
\label{fig:results-solved}
\end{figure}

\sparagraph{Benchmarks with non-single-invocation synthesis conjectures.}
Configuration {\bf cvc4+sg} is the only \cvc configuration that can
process benchmarks with synthesis conjectures that are not single-invocation.
The results for \esolver and {\bf cvc4+sg} on such benchmarks from
SyGuS 2014 are shown in Figure~\ref{fig:results-solved}.
Configuration {\bf cvc4+sg} solved 53 of them over a total of 67.
\esolver solved 58 and additionally reported that 6 had no solution.
In more detail, \esolver solved 7 benchmarks that {\bf cvc4+sg} did not, 
while {\bf cvc4+sg} solved 2 benchmarks (from the {\bf vctrl} class) 
that \esolver could not solve.
In terms of precision, {\bf cvc4+sg} is quite competitive with the state 
of the art on these benchmarks.
To give other points of comparison,
at the SyGuS 2014  competition~\cite{AlurETAL2014SyGuSMarktoberdorf}
the second best solver (the Stochastic solver) solved 40 of these benchmarks 
within a one hour limit and Sketch solved 23.

\sparagraph{Overall results.}
In total, over the entire SyGuS 2014 benchmark set, 155 benchmarks can be solved 
by a configuration of \cvc that, whenever possible, runs the methods 
for single-invocation properties described in Section~\ref{sec:refutation-based},
and otherwise runs the method described in Section~\ref{sec:syntax-guided}.
This number is 27 higher than the 128 benchmarks solved in total by \esolver.
Running both configuration {\bf cvc4+sg} and {\bf cvc4+si} in parallel\footnote{
\cvc has a \emph{portfolio} mode that allows it to run multiple configurations
at the same time.
} 
solves 156 benchmarks, 
indicating that \cvc is highly competitive with state-of-the-art tools 
for syntax guided synthesis.
\cvc's performance is noticeably better than \esolver 
on single-invocation properties, 
where our new quantifier instantiation techniques give it a distinct advantage.

%In particular, two benchmark classes ({\bf array} and {\bf int}) contain benchmarks that our approach
%can solve almost instantaneously, while 

\sparagraph{Competitive advantage on single-invocation properties in the presence of ite.}
We conclude by observing that for certain classes of benchmarks, 
configuration {\bf cvc4+si} scales significantly better than 
state-of-the-art synthesis tools.
Figure~\ref{fig:results-max} shows this in comparison with \esolver 
for the problem of synthesizing a function that computes the maximum of $n$
integer inputs.
As reported by Alur et al.~\cite{AlurETAL2014SyGuSMarktoberdorf},
no solver in the SyGuS 2014 competition was able to synthesize 
such a function for $n = 5$ within one hour.

For benchmarks from the {\bf array} class, whose solutions are loop-free 
programs that compute the first instance of an element in a sorted array,
the best reported solver for these in~\cite{AlurETAL2014SyGuSMarktoberdorf} was
Sketch, which solved a problem for an array of length 7 in approximately 30
minutes.\footnote{
These benchmarks, as contributed to the SyGuS benchmark set, use integer variables only; they were generated by expanding fixed-size arrays 
and contain no operations on arrays.
}
In contrast, {\bf cvc4+si} was able to reconstruct solutions for arrays of size 15 (the largest benchmark in the class)
in 0.3 seconds, and solved each of the benchmarks in the class but 8  
within 1 second.

\begin{figure}[t] %bh]
\centering \scriptsize
\begin{tabular}{|c|c|c|c|c|c|c|c|c|c|}
\hline                                                                
 $n$ &  2   & 3   & 4 & 5 & 6 & 7 & 8 & 9 & 10
\\
\hline
{\bf esolver} & 0.01  & \hfill 1377.10  & --  & --  & --  & --  & --  & --  & --
\\                                    
{\bf cvc4+si} & 0.01  & \hfill 0.02 & 0.03  & 0.05  & 0.1 & 0.3 & 1.6 & 8.9 & 81.5
\\                                                                                                    
\hline                                                                
\end{tabular}
\caption{Results for parametric benchmarks class encoding the maximum of $n$ integers.
The columns show the run time for \esolver and \cvc with a 3600s timeout.
\label{fig:results-max}}
\end{figure}

\section{Conclusion}

We have shown that SMT solvers, instead of just acting as subroutines 
for automated software synthesis tasks, can be instrumented 
to perform synthesis themselves. 
We have presented a few approaches for enabling SMT solvers 
to construct solutions for the broad class of syntax-guided synthesis problems
and discussed their implementation in \cvc.
This is, to the best of our knowledge, the first implementation of synthesis 
inside an SMT solver and it already shows considerable promise.
Using a novel quantifier instantiation technique and 
a solution enumeration technique for the theory of algebraic datatypes, 
our implementation is competitive with the state of the art represented 
by the systems that participated in the 2014 syntax-guided synthesis competition.
Moreover, for the important class of single-invocation problems 
when syntax restrictions permit the if-then-else operator, 
our implementation significantly outperforms those systems.

\sparagraph{Acknowledgments.}
We would like to thank Liana Hadarean for helpful discussions 
on the normal form used  in \cvc for bit vector terms.

%\end{document}

{
\bibliographystyle{abbrv}
\bibliography{main,managed}
}

\end{document}